\documentclass[runningheads,a4paper]{llncs}

\usepackage{amsmath}
\usepackage{amssymb}

\usepackage{algorithm}
\usepackage{subfig}

\usepackage{algpseudocode}

\usepackage{amsmath}
\usepackage{amssymb}
\usepackage{algorithm}

\usepackage{algpseudocode}
\usepackage{stfloats} 

\usepackage{multirow} 
\usepackage{booktabs}

\usepackage{relsize}

\usepackage{url}
\urldef{\mailsa}\path|{nikolako,|
\urldef{\mailsb}\path|garofala}@ceid.upatras.gr|    
\newcommand{\keywords}[1]{\par\addvspace\baselineskip
\noindent\keywordname\enspace\ignorespaces#1}

\usepackage{graphicx}

\usepackage{epsfig}

\usepackage{tikz}
\usetikzlibrary{arrows}

\begin{document}

\mainmatter  
\title{Random Surfing Without Teleportation}
\titlerunning{Random Surfing Without Teleportation}

\author{Athanasios N. Nikolakopoulos \and John D. Garofalakis}
\authorrunning{Nikolakopoulos and Garofalakis}
\institute{Computer Engineering and Informatics Department, University of Patras
\and CTI and Press ``Diophantus'' \\ \mailsa \mailsb}

\toctitle{Random Surfing Without Teleportation}
\tocauthor{Nikolakopoulos and Garofalakis}
\maketitle

\begin{abstract}
In the standard \textit{Random Surfer Model}, the teleportation matrix is necessary to ensure that the final PageRank vector is well-defined. The introduction of this matrix, however, results in serious problems and imposes fundamental limitations to the quality of the ranking vectors. 
In this work, building on the recently proposed \textit{NCDawareRank} framework, we exploit the decomposition of the underlying space into blocks, and we derive easy to check necessary and sufficient conditions for \textit{random surfing without teleportation}.
\end{abstract}

\section{Introduction \& Motivation}
The astonishing amount of information available on the Web 
and the highly variable quality of its content generate the need for an absolute measure of importance for Web pages, that can be used to improve the performance of Web search. Link Analysis algorithms such as the celebrated \textit{PageRank}, try to answer this need by using the link structure of the Web to assign authoritative weights to the pages~\cite{pagerank}.

\begin{figure}
	\centering
	\subfloat[][PageRank]{\epsfig{file=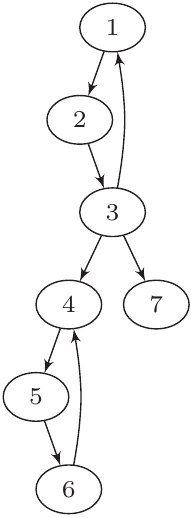}}
	\qquad \qquad
	\subfloat[][NCDawareRank]{\hspace*{2.0em}\epsfig{file=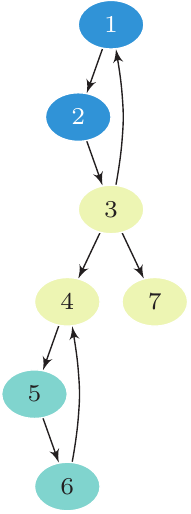}\hspace*{1.3em}}
	\caption{In the left figure we see a tiny graph as viewed by PageRank and in the right, the same graph as viewed by NCDawareRank. Same colored nodes  belong to the same block and are considered related according to a given criterion.}
	\label{tiny_web}
\end{figure}
PageRank's approach is based on the assumption that links convey human endorsement. For example,  the existence of 
a link from page $3$ to page $7$ in Fig.~\ref{tiny_web}(a) is seen as a testimonial of the importance of page $7$. Furthermore, the amount of importance conferred to page $7$ is proportional to 
the importance of page $3$ and inversely proportional to the number of pages 
$3$ links to. In their original paper, Page et al.~\cite{pagerank} imagined of a 
\textit{random surfer} who, with probability $\alpha$ follows the links of a Web page, and with probability $1-\alpha$ jumps to a different page uniformly at random. Then, following this metaphor, the overall importance of a page was defined to be equal to the fraction of time this random surfer spends on it, in the long run.

Formulating PageRank's basic idea with a mathematical model, involves viewing the Web as a directed graph with Web pages as vertices and hyperlinks as edges. Given this graph, we can construct a \textit{row-normalized hyperlink matrix} $\mathbf{H}$, whose element $[\mathbf{H}]_{uv}$ is one over the outdegree of $u$ if there is a link from $u$ to $v$, or zero otherwise. The matter of dangling nodes is fixed with some sort of stochasticity adjustment, thereby transforming the initial matrix $\mathbf{H}$, to a stochastic matrix. 

A second adjustment is needed to certify that the final matrix is irreducible and aperiodic, so that it possesses a unique positive stationary probability distribution. That is ensured by the introduction of the \textit{damping factor} $\alpha$ and a \textit{teleportation} matrix $\mathbf{E}$, usually defined by $\mathbf{E} = \frac{1}{n}\mathbf{e}\mathbf{e}^{\intercal}$. The resulting matrix is given by:
\begin{equation}
	\label{PageRank_Matrix_G}
	\mathbf{G} = \alpha \mathbf{H} + (1-\alpha)\mathbf{E} 
\end{equation}
PageRank vector is the unique stationary distribution of the Markov chain corresponding to matrix $\mathbf{G}$.

The choice of the damping factor has received much attention since it determines the fraction of the importance of a node that is propagated through the edges rather than scattered throughout the graph via the teleportation matrix. Obviously, picking a very small damping factor ignores the link structure of the graph and results in uninformative ranking vectors. On the other hand, setting the damping factor very close to one,  causes  a number of serious problems. From a computational perspective,  as $\alpha\to 1$, the number of iterations till convergence to the PageRank vector grows prohibitively, and also makes the computation of the rankings numerically ill-conditioned~\cite{kamvar2003condition,LangvilleMeyer06}. Moreover, from a qualitative point of view, various studies indicate that damping factors close to 1 
result into counterintuitive ranking vectors where  all the PageRank gets concentrated mostly in irrelevant nodes, while the Web's core component is assigned null rank~\cite{Avrachenkov:2007:DPM:1777879.1777881,DBLP:conf/dagstuhl/BoldiSV07,Boldi:2009:PFD:1629096.1629097,Nikolakopoulos:2013:NNR:2433396.2433415}. Finally, the very existence of the damping factor and the related teleportation matrix ``opens the door'' to direct manipulation of the ranking score through link spamming~\cite{constantine2009random,Eiron:2004:RWF:988672.988714}. 

In the literature there have been proposed several ranking methods that try to address these issues. Boldi~\cite{Boldi:2005:TRW:1062745.1062787} proposed an algorithm that eliminates PageRank's dependency on the arbitrarily chosen parameter $\alpha$ by integrating the ranking vector over the entire range of possible damping factors. Baeza-Yates et al.~\cite{Baeza-Yates:2006:GPD:1148170.1148225} introduced a family of link-based ranking algorithms parametrised by the selection of a damping function that describes how rapidly the importance of paths decays as the path length increases. Constantine and Gleich~\cite{constantine2009random} proposed a ranking method that considers the influence of a population of random surfers, each choosing its own damping factor from a probability distribution.

All the above methods attack the problem from the damping factor point of view, while taking the teleportation matrix as granted. Nikolakopoulos and Garofalakis~\cite{Nikolakopoulos:2013:NNR:2433396.2433415}, on the other hand, focus on the teleportation model itself. Building on the intuition behind Nearly Decomposable Systems~\cite{Courtois:1985:TSD:3812.3814,Simon:1996:SA:237774,simon1961aggregation}, the authors proposed \textit{NCDawareRank}; a novel ranking framework that generalizes and
refines PageRank by enriching the teleportation model in a computationally efficient way. NCDawareRank decomposes the underlying space into NCD blocks, and uses these blocks to define indirect relations between the nodes in the graph (Fig.~\ref{tiny_web}(b)) which lead to the introduction of a new inter-level proximity component. A comprehensive set of experiments done by the authors using real snapshots of the Web Graph showed that the introduction of this decomposition alleviates the negative effects of uniform teleportation and produces ranking vectors that display low sensitivity to sparsity and, at the same time, exhibit resistance to direct manipulation through link spamming (see the discussion in Sections 4.2 and 4.3 in~\cite{Nikolakopoulos:2013:NNR:2433396.2433415} for further details). However, albeit reducing some of its negative effects, NCDawareRank model also includes the standard teleportation matrix as a purely mathematical necessity. But, is it?

The main questions we try to address in this work are the following: 
\textit{Is it possible to discard the uniform teleportation altogether? And if so, under which conditions?} Thankfully, the answer is yes. In particular, we show that, the definition of the NCD blocks, can be enough to ensure the production of well-defined ranking vectors without resorting to uniform teleportation. The criterion for this to be true is expressed solely in term of properties of the proposed decomposition, which makes it very easy to check and at the same time gives insight that can lead to better decompositions for the particular ranking problems under consideration.  

The rest of the paper is organized as follows: After discussing NCDawareRank model (Section~\ref{Sec:NCDawareRank}) we derive sufficient and necessary conditions under which the inter-level proximity matrix enables us to discard the teleportation  matrix completely (Section~\ref{SubSec:Primitivity}). In Section~\ref{Sec_Overlapping}, we generalize NCDawareRank model, in order to allow the definition of overlapping blocks without compromising its theoretical and computational properties. Finally, in Section~\ref{Sec_Conclussions} we discuss future direction and conclude this work.

\section{NCDawareRank Model}
\label{Sec:NCDawareRank}
Before we proceed to our main result, we present here the basic definitions behind the NCDawareRank model. Our presentation follows the one given in~\cite{Nikolakopoulos:2013:NNR:2433396.2433415}.

\subsection{Notation}
All vectors are represented by bold lower case letters and they are column vectors (e.g., $\boldsymbol{\pi}$). All matrices are represented by bold upper case letters (e.g., $\mathbf{P} $). The $i^{\text{th}}$ row and $j^{\text{th}}$ column of matrix $\mathbf{P}$ are denoted $\mathbf{p}^\intercal_{i}$ and $\mathbf{p}_{j}$, respectively. The $ij^{th}$ element of matrix $\mathbf{P}$ is denoted $[\mathbf{P}]_{ij}$. We use $\operatorname{\textbf{Diag}}(\boldsymbol{\omega})$ to denote the matrix having vector $\boldsymbol{\omega}$ on its diagonal, and zeros elsewhere. We use calligraphic letters to denote sets (e.g., $\mathcal{U,V}$). $[1,n]$ is used to denote the set of integers $\{1,2,\dots,n\}$. Finally, symbol $\triangleq$ is used in definition statements.

\subsection{Definitions}

Let $\mathcal{U}$ be a set of nodes (e.g. the universe of Web pages) and denote $n\triangleq|\mathcal{U}|$.  Consider a node $u$ in $\mathcal{U}$. We denote $\mathcal{G}_u$ to  be the set of nodes that can be visited in a single step from $u$. Clearly, $d_u\triangleq|\mathcal{G}_u|$ is the out-degree of $u$, i.e. the number of outgoing edges of $u$.

We consider a partition of the underlying space $\mathcal{U}$ that defines a \textbf{decomposition}:
\begin{equation}
\mathcal{M} \triangleq \{\mathcal{D}_1,\dots,\mathcal{D}_K\}
\end{equation}
 such that, $\mathcal{D}_k\neq \emptyset$, for all $k$ in $[1,K]$.

Each set $\mathcal{D}_I$ is referred to as an \textbf{NCD Block}, and its elements are considered related according to a given criterion,  chosen for the particular ranking problem (e.g. the partition of the set of Web Pages into websites). 

We define $\mathcal{M}_u$ to be the set of \textit{proximal} nodes  of $u$, i.e the union of the NCD blocks that contain $u$ and the nodes it links to. Formally, the set $\mathcal{M}_u$ is defined by:

\begin{equation}
\mathcal{M}_u \triangleq \bigcup_{{w \in (u\cup\mathcal{G}_u)}}\mathcal{D}_{(w)}
\label{def:proximal}
\end{equation} 
where $\mathcal{D}_{(u)}$ is used to denote the unique block that includes node $u$.
Finally, $N_u$  denotes the number of different blocks in $\mathcal{M}_u$.

\begin{description}
	\item[Hyperlink Matrix.] The hyperlink matrix $\mathbf{H}$, as in the standard PageRank Model, is a row normalized version of the adjacency matrix induced by the  graph, and its $uv^{th}$ element is defined as follows:
	\begin{equation}
	[\mathbf{H}]_{uv} \triangleq \left\{
	\begin{array}{l l}
	\frac{1}{d_u} & \quad \mbox{if $v \in \mathcal{G}_u$}\\
	0 & \quad \mbox{otherwise}\\
	\end{array} \right. 
	\end{equation}
	
	Matrix $\mathbf{H}$ is assumed to be a row-stochastic matrix. The matter of dangling nodes (i.e. nodes with no outgoing edges) is considered fixed through some sort of stochasticity adjustment. 
	
	\item[Inter-Level Proximity Matrix.] The Inter-Level Proximity matrix $\mathbf{M}$ is created to depict the interlevel connections between the nodes in the graph.
	In particular, each row of matrix $\mathbf{M}$ denotes a probability vector $\mathbf{m}^\intercal_u$, that distributes evenly its mass between the $N_u$ blocks of $\mathcal{M}_u$, and then, uniformly to the included nodes of each block. Formally, the $uv^{th}$ element of matrix $\mathbf{M}$, that relates the node $u$ with node $v$, is defined as
	\begin{equation}
	[\mathbf{M}]_{uv}\triangleq \left\{
	\begin{array}{l l}
	\frac{1}{N_u|\mathcal{D}_{(v)}|} & \quad \mbox{if $v \in \mathcal{M}_u$}\\
	0 & \quad \mbox{otherwise}\\
	\end{array} \right. 
	\label{def:M}
	\end{equation}
	From the definition of the NCD blocks and the proximal sets, it is clear that whenever the number of blocks is smaller than the number of nodes in the graph, i.e. $K<n$, matrix $\mathbf{M}$ is necessarily low-rank; in fact, a closer look at the definitions~(\ref{def:proximal}) and~(\ref{def:M}) above, suggests that matrix $\mathbf{M}$ admits a very useful factorization, which was shown in~\cite{Nikolakopoulos:2013:NNR:2433396.2433415} to ensure the tractability of the resulting model. In particular, matrix 
	$\mathbf{M}$ can be expressed as a product of 2 extremely sparse matrices, 
	$\mathbf{R}$ and $\mathbf{A}$, defined below.

	Matrix $\mathbf{A} \in \mathbb{R}^{K\times n}$ is defined as follows:  
	\begin{equation}
	\mathbf{A} \triangleq
	\begin{bmatrix}
	\mathbf{e}^{\intercal}_{|\mathcal{D}_1|} & \boldsymbol{0} & \boldsymbol{0} & \cdots & \boldsymbol{0} \\
	\boldsymbol{0} & \mathbf{e}^{\intercal}_{|\mathcal{D}_2|} & \boldsymbol{0} & \cdots & \boldsymbol{0} \\
	\boldsymbol{0} & \boldsymbol{0} &  \mathbf{e}^{\intercal}_{|\mathcal{D}_3|} & \cdots & \boldsymbol{0} \\
	\vdots & \vdots & \vdots & \ddots & \boldsymbol{0} \\
	\boldsymbol{0} & \boldsymbol{0} & \boldsymbol{0} & \cdots & \mathbf{e}^{\intercal}_{|\mathcal{D}_K|}    
	\end{bmatrix} 
	\label{matrixA}
	\end{equation}
	where $\mathbf{e}^{\intercal}_{|\mathcal{D}_k|}$ denotes a row vector in 
	$\mathbb{R}^{|\mathcal{D}_k|}$ whose elements are all 1. Now, using the diagonal matrix $\mathbf{\Delta}$: 
	\begin{equation}
	\mathbf{\Delta} \triangleq \operatorname{\mathbf{Diag}}\left(\begin{bmatrix}|\mathcal{D}_1| & & & |\mathcal{D}_2| & &  & \cdots & &  & |\mathcal{D}_K|\end{bmatrix}\right)
	\label{matrixDelta}
	\end{equation} 
	and a row normalized matrix $\mathbf{\Gamma} \in \mathbb{R}^{n\times K}$, whose rows correspond  to nodes and columns to blocks and its elements are given by
	\begin{equation}
	[\mathbf{\Gamma}]_{ij} \triangleq \left\{
	\begin{array}{l l}
	\frac{1}{N_u} & \quad \mbox{if $\mathcal{D}_{j} \in \mathcal{M}_{u_i}$}\\
	0 & \quad \mbox{otherwise}\\
	\end{array} \right. 
	\label{matrixGamma}
	\end{equation}
	we can define the matrix $\mathbf{R}$ as follows:
	\begin{equation}
	\mathbf{R} \triangleq \mathbf{\Gamma}\mathbf{\Delta}^{-1}
	\label{matrixR}
	\end{equation}
	
	Using~(\ref{matrixA}) and~(\ref{matrixR}), it is straight forward to verify that: 
	\begin{eqnarray}
	\mathbf{M}  &=&  \mathbf{R} \mathbf{A} \\
	\mathbf{R} \in \mathbb{R}^{n\times K}& & \mathbf{A} \in  \mathbb{R}^{K\times n} \nonumber 
	\end{eqnarray}
	As pointed out by the authors~\cite{Nikolakopoulos:2013:NNR:2433396.2433415}, this factorization can lead to significant advantages in realistic scenarios, in terms of both storage and computability (see~\cite{Nikolakopoulos:2013:NNR:2433396.2433415}, Section 3.2.1). 
	\item[Teleportation Matrix.] Finally, NCDawareRank model also includes a teleportation matrix $\mathbf{E}$,
	\begin{equation}
	\mathbf{E} \triangleq \mathbf{e} \mathbf{v}^\intercal
	\end{equation} where, $\mathbf{v>0}$ such that $\mathbf{v^\intercal e} = 1$. The introduction of this matrix, can be seen as a remedy to ensure that the underlying Markov chain, corresponding to the final matrix, is irreducible and aperiodic and thus has a unique positive stationary probability distribution \cite{Nikolakopoulos:2013:NNR:2433396.2433415}.  
\end{description}

\bigskip

\noindent The resulting matrix which we denote $\mathbf{P}$ is expressed by:
\begin{equation}
	\label{NCDawareRank}
	\mathbf{P} = \eta \mathbf{H} + \mu \mathbf{M} + (1-\eta - \mu) \mathbf{E}  
\end{equation}
Parameter $\eta$ controls the fraction of importance delivered to the outgoing edges and parameter $\mu$ controls the fraction of importance that will be propagated to the proximal nodes. In order to ensure the irreducibility and aperiodicity of the final stochastic matrix in the general case, $\eta + \mu$ must be less than $1$.  This leaves $1-\eta-\mu$ of importance scattered throughout the graph through matrix $\mathbf{E}$.


\bigskip

\section{Necessary and Sufficient Conditions for Random Surfing Without Teleportation}

Although in the general case the teleportation matrix is required to ensure the final stochastic matrix produces a well-defined ranking vector, in this Section we show that NCDawareRank model carries the possibility of discarding matrix $\mathbf{E}$ altogether. Before we proceed to the proof of our main result (Section~\ref{SubSec:Primitivity}) we present here the  necessary preliminary definitions and theorems.

\subsection{Preliminaries}

\begin{definition}[Irreducibility] 
	An $n\times n$ non-negative matrix $\mathbf{P}$ is called \textit{irreducible} if for every pair of indices $i,j \in [1,n]$, there exists a positive integer $m \equiv m(i,j)$ such that $[\mathbf{P}^m]_{ij}>0$. The class of all non-negative irreducible matrices is denoted $\mathfrak{I}$.
\end{definition}

\begin{definition}[Period]
	The \textit{period} of an index $i\in[1,n]$ is defined to be the greatest common divisor of all positive integers $m$ such that $[\mathbf{P}^m]_{ii}>0$.
\end{definition}

\begin{proposition}[Periodicity as a Matrix Property]
	 For an irreducible matrix, the period of every index is the same and is referred to as the period of the matrix.
\end{proposition}

\begin{definition}[Primitivity]
	An irreducible matrix with period $d=1$, is called \textit{primitive}. The important subclass of all primitive matrices will be denoted $\mathfrak{P}$.
\end{definition}
Finally, we give here, without proof, the following fundamental result of the theory of non-negative matrices\footnote{For thorough treatment of the theory as well as proofs to several formulations of the Perron-Frobenius theorem the interested reader can see~\cite{seneta2006non}}. 
\begin{theorem}[Perron-Frobenius Theorem for Primitive Matrices\cite{Frobenius-1908-theorem,Perron-1907-theorem}]
	Suppose $\mathbf{T}$ is an $n\times n$ non-negative primitive matrix. Then, there exists an eigenvalue $r$ such that:
	\begin{enumerate}
		\item[{\upshape(a)}] $r$ is real and positive,
		\item[{\upshape(b)}] with $r$ can be associated strictly positive left and right eigenvectors,
		\item[{\upshape(c)}] $r>\lvert\lambda\rvert$ for any eigenvalue $\lambda\neq r$
		\item[{\upshape(d)}] the eigenvectors associated with $r$ are unique to constant multiples,
		\item[{\upshape(e)}] if $0\leq \mathbf{B} \leq \mathbf{T}$ and $\beta$ is an eigenvalue of $\mathbf{B}$, then $\lvert \beta \rvert \leq r$. Moreover, 
		\begin{displaymath}
		\lvert \beta \rvert = r \quad \Longrightarrow \quad \mathbf{B}=\mathbf{T}
		\end{displaymath}
		\item[{\upshape(f)}] $r$ is a simple root of the characteristic equation of $\mathbf{T}$.
	\end{enumerate}
\end{theorem}

\subsection{NCDawareRank Primitivity Criterion}
 \label{SubSec:Primitivity}
Mathematically, in the standard PageRank model the introduction of the teleportation matrix can be seen as a \textit{primitivity adjustment} of the final stochastic matrix. Indeed, the hyperlink matrix is typically reducible~\cite{LangvilleMeyer06,pagerank}, so if the teleportation matrix had not existed the PageRank vector would not be well-defined. 

In the general case, the same holds for NCDawareRank, as well. However, for suitable decompositions of the underlying graph, matrix $\mathbf{M}$ opens the door for achieving primitivity without resorting to the uninformative teleportation matrix.  Here, we show that this ``suitability'' of the decompositions can, in fact, be reflected on the properties of a low dimensional  \textbf{Indicator Matrix} defined below: 

\begin{definition}[Indicator Matrix]
	For every decomposition $\mathcal{M}$, we define an Indicator Matrix $\mathbf{W}\in \mathbb{R}^{K \times K}$ designed to capture the inter-block relations of the underlying graph. Concretely, matrix $\mathbf{W}$ is defined as follows:
	\begin{displaymath}
		\mathbf{W} \triangleq \mathbf{A}\mathbf{R},  
	\end{displaymath} 
	where $\mathbf{A,R}$ are the factors of the inter-level proximity matrix $\mathbf{M}$.
\end{definition}

Clearly, whenever $[\mathbf{W}]_{IJ}$ is positive, there exists a node $u \in \mathcal{D}_I$ such that $\mathcal{D}_J \in \mathcal{M}_u$. Intuitively, one can see that a positive element in matrix $\mathbf{W}$ implies the existence of possible inter-level ``random surfing paths'' between the nodes belonging to the corresponding blocks. Thus, if the indicator matrix $\mathbf{W}$ is irreducible, these paths exist between every pair of nodes in the graph, which makes the stochastic matrix $\mathbf{M}$ also irreducible. 

In fact, in the following theorem we show that the irreducibility of matrix $\mathbf{W}$ is enough to certify the primitivity of the final NCDawareRank matrix, $\mathbf{P}$. Then, just choosing positive numbers $\eta,\mu$ that sum to one, leads to a well-defined ranking vector produced by an NCDawareRank model without a teleportation component.

\begin{theorem}[Primitivity Criterion]
	The NCDawareRank matrix  $\mathbf{P} = \eta\mathbf{H} + \mu\mathbf{M}$, with $\eta$ and $\mu$ positive real numbers such that $\eta + \mu = 1$, is primitive if and only if the indicator matrix $\mathbf{W}$ is irreducible. Concretely, $\mathbf{P} \in \mathfrak{P} \iff \mathbf{W} \in \mathfrak{I}$.
	\label{Theorem:PrimitivityConditions}
\end{theorem}

\begin{proof}
	We will first prove that 
	\begin{equation}
	\mathbf{W} \in \mathfrak{I} \implies \mathbf{P} \in \mathfrak{P}
	\end{equation} 
	
	First notice that whenever matrix $\mathbf{W}$ is irreducible then it is also primitive. In particular, it is known that when a non-negative irreducible matrix has at least one positive diagonal element, then it is also primitive. In case of matrix $\mathbf{W}$, notice that by the definition of the proximal sets and matrices $\mathbf{A,R}$, we get that $[\mathbf{W}]_{ii}>0$ for every $i$ in $[1,K]$. Thus, the irreducibility of the indicator matrix ensures its primitivity also. Formally, we have 
	\begin{equation}
	\mathbf{W} \in \mathfrak{I} \implies \mathbf{W} \in \mathfrak{P}
	\end{equation}

	Now if the indicator matrix $\mathbf{W}$ is primitive the same is true for the inter-level proximity matrix $\mathbf{M}$. We prove this in the following lemma.
	
	\begin{lemma} The primitivity of the indicator matrix $\mathbf{W}$ implies the primitivity of the inter-level proximity matrix $\mathbf{M}$, defined over the same decomposition, i.e
	\begin{equation}
		\mathbf{W} \in \mathfrak{P} \implies \mathbf{M} \in \mathfrak{P}		
	\end{equation}
	\label{lemma:W2M}
	\end{lemma}
	\begin{proof}
	It suffices to show that there exists a number $m$, such that for every pair of indices $i,j$, $[\mathbf{M}^m]_{ij}>0$ holds. Or equivalently there exists a positive integer $m$ such that $\mathbf{M}^m$ is a positive matrix (see \cite{seneta2006non}).
	
	This can be seen easily using the factorization of matrix $\mathbf{M}$ given above. In particular, since $\mathbf{W}\in\mathfrak{P}$, there exists a positive integer $k$ such that $\mathbf{W}^k>0$. Now, if we choose $m = k+1$, we get:
	\begin{eqnarray}
		\mathbf{M}^m & = & (\mathbf{RA})^{k+1} \nonumber \\
		& = & \underbrace{\mathbf{(RA)(RA)\cdots (RA)}}_{k+1\text{ times}} \nonumber \\
		& = & \mathbf{R} \underbrace{\mathbf{(AR)(AR)\cdots (AR)}}_{k\text{ times}}\mathbf{A} \nonumber \\
		& = & \mathbf{R} \mathbf{W}^k \mathbf{A}
		\label{rel:M_Prim}
	\end{eqnarray}
	
	However, matrix $\mathbf{W}^k$ is positive and since every row of matrix $\mathbf{R}$ and every column of matrix $\mathbf{A}$ are -- by definition -- non-zero, the final matrix, $\mathbf{M}^m$, is also positive. Thus, $\mathbf{M} \in \mathfrak{P}$, and the proof is complete. \qed
	\end{proof}
	
	Now, in order to get the primitivity of the final stochastic matrix $\mathbf{P}$, we use the following useful lemma which shows that any convex combination of stochastic matrices that contains at least one primitive matrix, is also primitive.
	\begin{lemma}
		Let $\mathbf{A}$ be a primitive stochastic matrix and $\mathbf{B_1,B_2,\dots,B_n}$ stochastic matrices, then matrix
		\begin{displaymath}
		\mathbf{C} = \alpha \mathbf{A}+\beta_1\mathbf{B_1}+\dots+\beta_n\mathbf{B_n}
		\end{displaymath} where $\alpha>0$ and $\beta_1,\dots,\beta_n\geq0$ such that $\alpha+\beta_1+\dots+\beta_n=1$ is a primitive stochastic matrix.
		\label{Lemma1}
	\end{lemma}

	\begin{proof}
		Clearly matrix $\mathbf{C}$ is stochastic as a convex combination of stochastic matrices (see~\cite{horn2012matrix}).
		For the primitivity part it suffices to show that there exists a natural number, $m$, such that $\mathbf{C}^m>0$. This can be seen very easily. In particular, since matrix $\mathbf{A} \in \mathfrak{P}$, there exists a number $k$ such that every element in $\mathbf{A}^{k}$ is positive.
		
		Consider the matrix $\mathbf{C}^m$:
		\begin{eqnarray}
			\mathbf{C}^m & = & (\alpha\mathbf{A}+\beta_1\mathbf{B_1}+\dots+\beta_n\mathbf{B_n})^m \nonumber \\
			& = & \alpha^m\mathbf{A}^m + (\text{sum of non-negative matrices})
		\end{eqnarray}
		Now letting $m=k$, we get that every element of matrix $\mathbf{C}^{k}$ is strictly positive, which completes the proof.   \qed
	\end{proof}

	As we have seen, when $\mathbf{W}\in \mathfrak{I}$, matrix $\mathbf{M}$ is primitive. Furthermore, $\mathbf{M}$ and $\mathbf{H}$ are by definition stochastic. Thus, Lemma~\ref{Lemma1} applies and we get that the NCDawareRank  matrix $\mathbf{P}$, is also primitive. 
	In conclusion, we have shown that: 
	\begin{equation}
		\mathbf{W}\in\mathfrak{I} 
		\implies\mathbf{W}\in\mathfrak{P}\implies\mathbf{M}\in\mathfrak{P}\implies\mathbf{P}\in\mathfrak{P}
		\label{ReverseProof}
	\end{equation}
	which proves the reverse direction of the theorem.
	
	To prove the forward direction (i.e. $\mathbf{P} \in \mathfrak{P} \implies \mathbf{W} \in \mathfrak{I}$) it suffices to show that whenever matrix $\mathbf{W}$ is reducible, matrix $\mathbf{P}$ is also reducible (and thus, not primitive \cite{seneta2006non}). 	First observe that when matrix $\mathbf{W}$ is reducible the same holds for matrix $\mathbf{M}$. 
	
	\begin{lemma}
		The reducibility of the indicator matrix $\mathbf{W}$ implies the reducibility of the inter-level proximity matrix $\mathbf{M}$. Concretely,
		\begin{equation}
		\mathbf{W} \notin \mathfrak{I} \implies \mathbf{M} \notin \mathfrak{I}
		\end{equation}
	\end{lemma}
	
	\begin{proof}
	Assume that matrix $\mathbf{W}$ is reducible. Then, there exists a permutation matrix $\mathbf{\Pi}$ such that $\mathbf{\Pi W \Pi^\intercal}$ has the form 
	\begin{equation}
		\begin{bmatrix}
			\mathbf{X} & \mathbf{Z} \\
			\mathbf{0} & \mathbf{Y}
		\end{bmatrix} 
		\label{rel:BlockUpperDiagonal}
	\end{equation}
	where $\mathbf{X,Y}$ are square matrices~\cite{seneta2006non}. Notice that a similar block upper triangular form can be then achieved for matrix $\mathbf{M}$. In particular, the existence of the block zero matrix in~(\ref{rel:BlockUpperDiagonal}), together with the definition of matrices $\mathbf{A,R}$ ensures the existence of a set of blocks, that have the property none of their including nodes to have outgoing edges to the rest of the nodes in the graph\footnote{notice that if this was not the case, there would be a nonzero element in the block below the diagonal necessarily.}. Thus, organizing the rows and columns of matrix $\mathbf{M}$ such that these nodes are assigned the last indices, results in a matrix $\mathbf{M}$ that has a similarly block upper triangular form. This makes $\mathbf{M}$ reducible too.  \qed
	\end{proof}
	
	Thus, we only need to show that the reducibility of matrix $\mathbf{M}$ implies the reducibility of matrix $\mathbf{P}$ also. This can arise from the fact that by definition 
	\begin{equation}
		[\mathbf{M}]_{ij}=0 \implies [\mathbf{H}]_{ij}=0.
	\end{equation} So, the permutation matrix that brings $\mathbf{M}$ in the form of~(\ref{rel:BlockUpperDiagonal}), has exactly the same effect on matrix $\mathbf{H}$. Similarly the  final stochastic matrix $\mathbf{P}$ has the same block upper triangular form as a sum of matrices $\mathbf{H}$ and $\mathbf{M}$. This makes matrix $\mathbf{P}$ reducible and hence non-primitive.
	
	Therefore, we have shown that $\mathbf{W} \notin \mathfrak{P} \implies \mathbf{P} \notin \mathfrak{I}$, which is equivalent to  
	\begin{equation}
		\mathbf{P} \in \mathfrak{P} \implies \mathbf{W} \in \mathfrak{I}
		\label{ForwardProof}
	\end{equation}
	Putting everything together, we see that both directions of our theorem have been established. Thus we get,
	\begin{equation}
	\mathbf{P} \in \mathfrak{P} \iff \mathbf{W} \in \mathfrak{I}
	\end{equation} and our proof is complete. \qed
\end{proof}

Now, when the stochastic matrix $\mathbf{P}$ is primitive, from the Perron-Frobenius theorem it follows that its largest eigenvalue -- which is equal to 1 -- is unique and it can be associated with strictly positive left and right eigenvectors. Therefore, under the conditions of Theorem~\ref{Theorem:PrimitivityConditions}, the ranking vector produced by the NCDawareRank model -- which is defined to be the stationary distribution of the stochastic matrix $\mathbf{P}$: (a) is uniquely determined as the (normalized) left eigenvector of $\mathbf{P}$ that corresponds to the eigenvalue 1 and, (b) its support includes every node in the underlying graph. The following corollary, summarizes the result.

\begin{corollary}
	When the indicator matrix $\mathbf{W}$ is irreducible, the ranking vector produced by NCDawareRank with $\mathbf{P} = \eta\mathbf{H} + \mu\mathbf{M}$, where $\eta,\mu$ positive real numbers such that $\eta + \mu = 1$ holds, denotes a well-defined distribution that assigns positive ranking to every node in the graph.
\end{corollary}

\section{Generalizing the NCDawareRank Model}
\label{Sec_Overlapping}
\subsection{The Case of Overlapping Blocks}
In our discussion so far, we assumed that the block decomposition defines a partition of the underlying space. However, in many realistic ranking scenarios it would be useful to be able to  allow the blocks to overlap. For example, if one wants to produce top N lists of movies for a ranking-based recommender system, using NCDawareRank, a very intuitive criterion for decomposition would be the one depicting the categorization of movies into genres~\cite{NikolakopoulosG14}. Of course, such a decomposition naturally results in  overlapping blocks, since a movie usually belongs to more than one genres.

Fortunately, the factorization of the inter-level proximity matrix, paves the path towards a straight forward generalization, that inherits all the useful mathematical properties and computational characteristics of the standard NCDawareRank model. 

In particular, it suffices to modify the definition of decompositions as  indexed families of non-empty sets
\begin{equation}
\mathcal{\hat{M}} \triangleq \{\mathcal{\hat{D}}_1,\dots,\mathcal{\hat{D}}_K\}
\end{equation} 
that collectively cover the underlying space, i.e. 
\begin{equation}\mathcal{U}=\bigcup_{k=1}^{K}\mathcal{\hat{D}}_k
\end{equation}
and to change slightly the definitions of the:
\begin{itemize}
	\item Proximal Sets: \begin{equation}
		\mathcal{\hat{M}}_u \triangleq \bigcup_{{w \in (u\cup\mathcal{G}_u),w \in \mathcal{\hat{D}}_k}}\mathcal{\hat{D}}_k
		\label{def:proximal_ovelapping}
	\end{equation} 
	\item Inter-Level Proximity Matrix: \begin{equation}
	[\mathbf{\hat{M}}]_{uv}\triangleq \sum_{\mathcal{\hat{D}}_k \in \mathcal{\hat{M}}_{u}, v \in \mathcal{\hat{D}}_k}\frac{1}{N_{u}\lvert \mathcal{\hat{D}}_k\rvert}
	\label{def:M_overlapping}
	\end{equation}
	\item Factor Matrices $\mathbf{\hat{A}},\mathbf{\hat{R}}$:   We first define a matrix $\mathbf{X}$,  whose $ik^{\textit{th}}$ element is 1, if $\mathcal{\hat{D}}_k \in \mathcal{\hat{M}}_i$ and zero otherwise, and a matrix $\mathbf{Y}\in \mathbb{R}^{K\times n}$, whose $kj^{\textit{th}}$ element is 1 if $v_j \in \mathcal{\hat{D}}_k$ and zero otherwise.
	Then, if $\mathbf{\hat{R}}$, $\mathbf{\hat{A}}$ denote the row-normalized versions of $\mathbf{X}$ and $\mathbf{Y}$ respectively, matrix $\mathbf{\hat{M}}$ can be expressed as: 
	\begin{equation}
	\mathbf{\hat{M}}  =  \mathbf{\hat{R}} \mathbf{\hat{A}}, \quad \mathbf{\hat{R}} \in \mathbb{R}^{n\times K}, \mathbf{\hat{A}} \in  \mathbb{R}^{K\times n}.
	\end{equation}
	
\end{itemize} 

\begin{remark}
	Notice that the Inter-Level Proximity Matrix above is a well-defined stochastic matrix, for every possible decomposition. Its stochasticity can arise immediately from the row normalization of matrices $\mathbf{\hat{R}}, \mathbf{\hat{A}}$, together with the fact that neither matrix $\mathbf{X}$ nor matrix $\mathbf{Y}$  have zero rows. Indeed, the existence of a zero row in matrix $\mathbf{X}$ implies 
		\begin{math}
			\mathcal{U} \neq \bigcup_{k=1}^{K}\mathcal{\hat{D}}_k,
		\end{math}
		which contradicts the definition of $\mathcal{\hat{M}}$; similarly the existence of a zero row in matrix $\mathbf{Y}$ contradicts the definition of the NCD blocks $\mathcal{\hat{D}}$ which are defined to be non-empty.
	
\end{remark}

\begin{remark}
Also notice that our primitivity criterion given by Theorem~\ref{Theorem:PrimitivityConditions}, applies in the overlapping case too, since our proof made no assumption for mutual exclusiveness for the NCD-blocks. In fact, it is intuitively evident that overlapping blocks promote the irreducibility of the indicator matrix $\mathbf{W}$. 
\end{remark}

\section{Discussion and Future Work}
\label{Sec_Conclussions}
In this work, using an approach based on the theory of non-negative matrices,
we study NCDawareRank's inter-level proximity model and we derive necessary and sufficient conditions, under which the underlying  decomposition alone could result in a well-defined ranking vector -- eliminating the need for uniform teleportation. Our goals here were mainly theoretical. However, our first findings in applying this ``no teleportation'' approach in realistic problems suggest that the conditions for primitivity are not prohibitively restrictive, especially if the criterion behind the definition of the decomposition implies overlapping blocks~\cite{NikolakopoulosG14,Nikolakopoulos2015,nikolakopoulos2015top}. 

A very exciting direction we are currently pursuing involves the spectral implications of the absence of the teleportation matrix. In particular, a very interesting problem would be to determine bounds of the subdominant eigenvalue of the stochastic matrix $\mathbf{P} = \eta\mathbf{H} + \mu\mathbf{M}$, when the indicator matrix $\mathbf{W}$ is irreducible. Another important direction would be to proceed to randomized definitions of blocks that satisfy the primitivity criterion and to test the effect on the quality of the ranking vector. 

In conclusion, we believe that our results, suggest that the NCDawareRank model presents a promising approach towards generalizing and enriching the standard random surfer model, and also carries the potential of providing an intuitive alternative teleportation scheme to the many applications of PageRank in hierarchical or otherwise specially structured graphs.

\bibliographystyle{splncs03}

\end{document}